\documentclass[11pt]{amsart}
\usepackage{soul,times,amssymb,amsthm,amsmath,amsfonts,bbm,mathrsfs,xcolor,subfigure,graphicx,booktabs,euscript,nicefrac,optidef}
\usepackage[shortlabels]{enumitem}
\allowdisplaybreaks
\usepackage{graphicx}
\usepackage[unicode]{hyperref}
\hypersetup{
    colorlinks,
    linkcolor={blue!80!black},
    citecolor={blue!80!black},
    urlcolor={blue!80!black}
}
\usepackage[normalem]{ulem}
\usepackage[font=footnotesize]{caption}
\usepackage[utf8]{inputenc}
\usepackage[T1]{fontenc}

\newtheorem{theorem}{Theorem}[section]
\newtheorem{lemma}[theorem]{Lemma}

\newtheorem{corollary}[theorem]{Corollary}
\theoremstyle{remark}
\newtheorem{remark}{Remark}
\newtheorem{definition}{Definition}[section]

\usepackage[top=1.4in,bottom=1.4in,left=1.4in,right=1.4in]{geometry}

\newcommand\nc\newcommand
\nc\bfa{{\boldsymbol a}}\nc\bfA{{\boldsymbol A}}\nc\cA{{\EuScript A}}
\nc\bfb{{\boldsymbol b}}\nc\bfB{{\boldsymbol B}}\nc\cB{{\mathscr B}}
\nc\bfc{{\boldsymbol c}}\nc\bfC{{\boldsymbol C}}\nc\cC{{\mathscr C}}
\nc\bfd{{\boldsymbol d}}\nc\bfD{{\boldsymbol D}}\nc\cD{{\EuScript D}}
\nc\bfe{{\boldsymbol e}}\nc\bfE{{\boldsymbol E}}\nc\cE{{\EuScript E}}
\nc\bff{{\boldsymbol f}}\nc\bfF{{\boldsymbol F}}\nc\cF{{\mathscr F}}
\nc\bfg{{\boldsymbol g}}\nc\bfG{{\boldsymbol G}}\nc\cG{{\EuScript G}}
\nc\bfh{{\boldsymbol h}}\nc\bfH{{\boldsymbol H}}\nc\cH{{\mathcal H}}
\nc\bfi{{\boldsymbol i}}\nc\bfI{{\boldsymbol I}}\nc\cI{{\mathcal I}}
\nc\bfj{{\boldsymbol j}}\nc\bfJ{{\boldsymbol J}}\nc\cJ{{\EuScript J}}
\nc\bfk{{\boldsymbol k}}\nc\bfK{{\boldsymbol K}}\nc\cK{{\EuScript K}}
\nc\bfl{{\boldsymbol l}}\nc\bfL{{\boldsymbol L}}\nc\cL{{\EuScript L}}
\nc\bfm{{\boldsymbol m}}\nc\bfM{{\boldsymbol M}}\nc\cM{{\EuScript M}}
\nc\bfn{{\boldsymbol n}}\nc\bfN{{\boldsymbol N}}\nc\cN{{\EuScript N}}
\nc\bfo{{\boldsymbol o}}\nc\bfO{{\boldsymbol O}}\nc\cO{{\EuScript O}}
\nc\bfp{{\boldsymbol p}}\nc\bfP{{\boldsymbol P}}\nc\cP{{\EuScript P}}
\nc\bfq{{\boldsymbol q}}\nc\bfQ{{\boldsymbol Q}}\nc\cQ{{\EuScript Q}}
\nc\bfr{{\boldsymbol r}}\nc\bfR{{\boldsymbol R}}\nc\cR{{\EuScript R}}
\nc\bfs{{\boldsymbol s}}\nc\bfS{{\boldsymbol S}}\nc\cS{{\EuScript S}}
\nc\bft{{\boldsymbol t}}\nc\bfT{{\boldsymbol T}}\nc\cT{{\EuScript T}}
\nc\bfu{{\boldsymbol u}}\nc\bfU{{\boldsymbol U}}\nc\cU{{\EuScript U}}
\nc\bfv{{\boldsymbol v}}\nc\bfV{{\boldsymbol V}}\nc\cV{{\mathscr V}}
\nc\bfw{{\boldsymbol w}}\nc\bfW{{\boldsymbol W}}\nc\cW{{\mathscr W}}
\nc\bfx{{\boldsymbol x}}\nc\bfX{{\boldsymbol X}}\nc\cX{{\EuScript X}}
\nc\bfy{{\boldsymbol y}}\nc\bfY{{\boldsymbol Y}}\nc\cY{{\mathscr Y}}
\nc\bfz{{\boldsymbol z}}\nc\bfZ{{\boldsymbol Z}}\nc\cZ{{\EuScript Z}}
\nc{\1}{{\mathbbm{1}}}
\nc\rr{{\mathbb R}}
\nc\ee{{\mathbb E}}
\nc\sS{{\mathcal S}}
\nc{\integers}{{\mathbb Z}}
\nc{\Z}{{\mathbb Z}}
\nc{\ff}{{\mathbb F}}
\nc{\ii}{{\mathbb I}}
\nc{\sC}{{\mathfrak C}}
\nc{\sL}{{\mathfrak L}}
\nc\hH{{\mathsf H}}
\nc\gG{{\mathsf G}}

\nc{\remove}[1]{}

\DeclareSymbolFont{bbold}{U}{bbold}{m}{n}
\DeclareSymbolFontAlphabet{\mathbbold}{bbold}

\DeclareMathOperator{\wdp}{WDP}

\DeclareMathOperator{\lpn}{LPN}

\DeclareMathOperator*{\argmin}{\arg\!\min}
\DeclareMathOperator{\negl}{negl}

\DeclareMathOperator{\bias}{bias}
\DeclareMathOperator{\poly}{poly}
\DeclareMathOperator{\Ber}{\rm Ber}
\nc{\dtv}{d_{\text{\rm TV}}}

\nc{\ion}{\{0,1,\dots,n\}}

\newcommand{\genstirlingI}[3]{%
  \genfrac{[}{]}{0pt}{#1}{#2}{#3}%
}

\newcommand{\stirlingI}[2]{\genstirlingI{}{#1}{#2}}

\allowdisplaybreaks	

\usepackage{xargs}
\usepackage[colorinlistoftodos,prependcaption]{todonotes}

\newcommandx{\rednote}[2][1=]{\todo[linecolor=red,backgroundcolor=red!25,bordercolor=red,#1]{#2}}
\newcommandx{\bluenote}[2][1=]{\todo[linecolor=blue,backgroundcolor=blue!25,bordercolor=blue,#1]{#2}}
\newcommandx{\yellownote}[2][1=]{\todo[linecolor=yellow,backgroundcolor=yellow!25,bordercolor=yellow,#1]{#2}}
\newcommandx{\greennote}[2][1=]{\todo[inline,linecolor=olive,backgroundcolor=green!25,bordercolor=olive,#1]{#2}}

\newcommand{\new}[1]{{\color{black} #1}}

\usepackage[normalem]{ulem}
\newcommand\redout{\bgroup\markoverwith{\textcolor{red}{\rule[0.5ex]{2pt}{0.8pt}}}\ULon}

\title[]{Limitations of the decoding-to-LPN reduction via code smoothing}

\author[]{Madhura Pathegama}\thanks{The authors are with the Department of ECE and Institute for Systems Research, University of Maryland, College Park, MD 20742. Emails: \{pankajap,abarg\}@umd.edu.}
\author[]{Alexander Barg}

\begin{document}

\begin{abstract} The Learning Parity with Noise (LPN) problem underlines several classic cryptographic primitives. Researchers have attempted to demonstrate the algorithmic hardness of this problem by finding reductions from the decoding problem of linear codes, for which several hardness results exist. Earlier studies used code smoothing as a tool to achieve reductions for codes with vanishing rate. This has left open the question of attaining a reduction with positive-rate codes. Addressing this case, we characterize the efficiency of the reduction in 
terms of the parameters of the decoding and LPN problems. As a conclusion, we isolate the parameter regimes
for which a meaningful reduction is possible and the regimes for which its existence is unlikely.
\end{abstract}

\maketitle

\section{Introduction}

The Learning Parity with Noise (LPN) problem underlies several classic cryptographic primitives, including symmetric 
encryption \cite{hopper2001secure,juels2005authenticating},  public-key cryptography \cite{alekhnovich2003more}, and 
collision-resistant hashing \cite{brakerski2019worst,yu2019collision}. 
For this reason, the problem of computational complexity of LPN has attracted
much attention in the literature. One line of work to show the hardness of this problem suggests to construct a reduction from the decoding problem of linear codes, arguing therefore that LPN is at least as hard as 
decoding generic codes. Reductions of this kind known to date
\cite{brakerski2019worst,yu2021smoothing,debris2022worst} rely on the technique known as code smoothing. To 
establish their claims, the authors of these works had to make certain assumptions about the parameters of 
code sequences involved in the constructions. Specifically, meaningful reductions have 
been established only for code sequences of asymptotically vanishing rate, whereas the regime of positive rates so far has proved beyond reach. While recent works have focused on improved reductions, in this paper we 
address a different possibility, namely that such reductions are not possible. 

The LPN problem was proposed in \cite{blum1993cryptographic} in a cryptographic context and was the subject of many subsequent works. Algorithms for solving it were designed in \cite{blum2003noise,lyubashevsky2005parity} among others. In the search (or computational) version of the LPN problem, the goal is to recover a uniform random secret $m\in \ff_2^k$ with high probability from ${N}$ samples of the form $(a_i, a_i^\intercal m +b_i)_{i=1}^{N},$ where the vectors $a_i$ are sampled uniformly and independently from $\ff_2^k$ and $b_i\sim\Ber(\delta)$ are independent Bernoulli random bits (we use column vectors throughout). We call the parameter $\delta$ the {\em noise rate} of the corresponding LPN problem. 

Studies into the hardness of LPN are motivated by an analogous computational task in the domain of lattice
cryptography called {\em Learning with Errors} (LWE)  \cite{regev2005lattices, micciancio2007worst}. In this
problem, one wants to recover a secret $m$ from a set of samples of the form $(a_i, a_i^{\intercal}m +b_i)$,
where the vectors $(a_i)_i$ are sampled from a large enough finite field $\ff_q^n$, while $(b_i)_i$ are sampled from a discrete Gaussian distribution. The primary technique employed in reducing worst-case lattice problems to LWE is lattice smoothing \cite{micciancio2007worst,debris2023smoothing}. By analogy, it is reasonable to expect that the worst-case decoding problems could similarly be reduced to the average-case LPN problem based on code smoothing, thereby supporting the hardness assumption.

Code smoothing can be informally described as follows. Given a code $\cC\subset \ff_2^n$, we define a probability distribution $P_\cC$ on $\ff_2^n$ by normalizing the indicator function $\1_\cC$. Applying a noise operator $T$ to $P_{\cC}$, we obtain
another distribution. We say that $T$ smooths the code $\cC$ if $TP_\cC$ is close to the uniform distribution on $\ff_2^n$. A more general version of this problem seeks to approximate any given ``output'' probability distribution $P$ and is known in information theory as channel resolvability \cite{han1993approximation}, with numerous applications in information-theoretic problems. We refer to \cite{pathegama2023smoothing} for an overview of the literature as well as some recent results on code smoothing outside the context of this paper.

For a given $[n,k,d]$ linear code $\cC$ with generator matrix $\gG\in \ff_2^{k\times n}$ and a nonnegative integer $w\leq n$, 
the worst-case decoding problem involves identifying the message $m$ from a noisy codeword $\gG^{\intercal}m +e$, with the promise that the Hamming weight of the error satisfies $|e|= w$. The hardness of this problem was discussed in \cite{downey1999parametrized}, and it was further shown in \cite{dumer2003hardness} it is NP-hard under the condition $w > d/2$. The non-promise version of the decoding problem is NP-hard even if one seeks to approximate the exact solution \cite{arora1997hardness}. 

Before proceeding, let us make our expectations of the decoding-to-LPN reduction more explicit. Clearly, if the noise rate $\delta=1/2,$ the LPN problem is not solvable. At the same time, symmetric cryptography
becomes possible if $\delta=1/2-1/\text{poly}(k)$ \cite{brakerski2019worst}, so we call a reduction
{\em meaningful} if it yields an LPN problem with this or smaller noise rate.

The general idea behind the reduction of decoding to LPN using smoothing is as follows. For a given linear code 
$\cC$ with generator matrix $\gG$, and an unknown error vector $e$ of weight  $w$, the goal is to find a 
distribution $P$ on $\ff_2^n$ so that sampling vectors $Z$ out of $P$ yields nearly uniform vectors $\gG Z$, 
while $e^{\intercal}Z$ is a biased Bernoulli random variable that is nearly independent of $\gG Z$. 
As will be seen below in the main text, finding such a distribution $P$ is equivalent to the smoothing.
Previous works employed certain explicit distributions for generating the random vectors $Z$ and used Fourier-analytic 
arguments or the leftover hash lemma \cite{impagliazzo1989pseudo} to demonstrate the effectiveness of these 
choices.

 The first paper devoted to the decoding-to-LPN reduction problem was the work of
Brakersky \emph{et al.} \cite{brakerski2019worst}, followed by Yu and Zhang \cite{yu2021smoothing}. Using code smoothing, these papers show that LPN is at least as hard as solving a worst-case decoding problem. To claim this result, their authors have to make certain assumptions (explicit or implicit) that lead to limitations on the parameters of codes in the decoding part of the reduction. In particular, the
reductions in \cite{brakerski2019worst,yu2021smoothing} become efficient once we require that the
rate $k/n$ of the codes vanish as $n$ increases. Additionally
the results in \cite{brakerski2019worst} rely on the assumption that the Hamming weights of the non-zero codewords are located between $l$ and $n-l$ for some $l$ (the authors call such codes {\em balanced}). The results of \cite{yu2021smoothing} require either balanced codes or certain independence conditions for the coordinates of the code. More recently, Debris-Alazard and Resch \cite{debris2022worst} presented a more general framework for utilizing code smoothing to reduce the decoding problem to the LPN problem. Their proof relies upon the work in \cite{debris2023smoothing} on lattice and code smoothing and does not involve any
restrictions on the weights of the code, but their reduction is also limited by the vanishing rate assumption.

It remained unclear whether the limitation of zero rate stemmed from the particular smoothing distribution 
used in the cited works or from inherent issues within the smoothing approach itself. For example, in all 
prior studies \cite{brakerski2019worst,yu2021smoothing,debris2022worst}, the authors employed a single 
universal smoothing distribution for all codes with the same parameters. However, in the reduction, this
limitation is not 
essential; instead, distributions can vary depending on the specific codes being considered. Addressing this problem, we prove that for constant-rate codes, achieving a meaningful reduction with constant relative error weight $w/n$ is impossible regardless of the smoothing distribution. This result relies on the following qualification: we require that the vectors $\gG Z$ (syndromes of the dual code) closely approximate uniformity, more formally, that the distance between their distribution and the uniform distribution on the space declines faster than $\poly(k)$. Allowing a slower approach to uniformity makes the reduction possible, although the parameter regime associated with it
degrades the performance of the decoder implied by the LPN solver; see Sec.~\ref{sec:reduction} for a more detailed discussion. 

Our results can be interpreted as follows. The potential reduction depends on three parameters, namely
the success probability of the LPN solver $\alpha$, the sample complexity of the solver $N$, and the
smoothing parameter $\varepsilon.$ The success probability of the decoding task accomplished via the smoothing reduction is then expressed as $\alpha-N\varepsilon$.
If the sample complexity scales as $N=2^{n^c}, 0<c<1$, then $\varepsilon$ has to be small to support the
decoding. In this regime, we show that meaningful reductions for constant-rate codes are impossible. On the other hand, if $N$ has slower growth, then a reduction is possible, but the decoder performance
is degraded.

A line of research closely related to the LPN problem addresses average-case decoding of linear codes. Two 
variants of the average-case decoding problem have been studied in the literature 
\cite{debris2022worst,bombar2023pseudorandomness}. Given a random linear code and a fixed error weight, the {\em 
search version} of the average-case decoding problem seeks to recover the original codeword from a given noisy 
codeword with high probability. The {\em decision version} of this problem seeks to resolve whether a given 
vector is drawn from a uniform distribution or from a distribution that results from adding noise to a uniform 
random codeword. The authors of \cite{debris2022worst} reduced the worst-case decoding problem to the search 
version of the average-case decoding problem. Similarly, in \cite{bombar2023pseudorandomness}, a reduction was 
obtained from the worst-case decoding problem to the decision version of the average-case decoding problem. 
In both cases, a crucial step was to reduce the worst-case decoding problem to the LPN problem, with code 
smoothing playing a pivotal role. Therefore, both approaches achieved meaningful reductions only for zero-rate 
codes. Since both results depend on reductions from worst-case decoding to LPN via code smoothing, our work suggests that generalizing them to positive-rate codes faces conceptual challenges and is unlikely.

\section{The worst-case decoding problem and the LPN problem}

\subsection{Notation}\label{sec:notation}
We consider linear codes $\cC\subset \ff_2^n$ of length $n$. For a random vector $Z$ on $\ff_2^n$ 
we denote by $P_Z$ its probability mass function (pmf). Given a pmf $P$ on
$\ff_2^n$, we write $Z\sim P$ to indicate that $Z$ is distributed according to $P$; hence $P_Z=P$. 
In particular, $P_{U_n}$ denotes the uniform distribution on $\ff_2^n$, where $U_n$ is the uniform random vector on $\ff_2^n$. For a code $\cC \subset \ff_2^n$,  $P_{\cC}$ denotes the uniform distribution on it and $X_{\cC}$ is a uniform random codeword of $\cC$. 
The notation $P_{\Ber(\delta)}$ refers to the pmf of a Bernoulli random variable $\xi$ with $P_\xi(1)=\delta$. \new{We often use the quantity
  $$
  \bias(\xi):=1/2-\delta
  $$
that measures the deviation of $\xi$ from a uniform Bernoulli random variable.} 

Given two probability measures $P$ and $Q$ defined on $\ff_2^n$,
define the {\em total variation distance} (TV-distance) as 
\begin{align}
    \dtv(P,Q) = \max_{A\subset \ff_2^n}|P(A)-Q(A)|. \label{eq: dtv-max}
\end{align}
It is well known that $\dtv(P,Q)$ can be written as
\begin{align}
    \dtv(P,Q) = \frac{1}{2}\sum_{x \in \ff_2^n}|P(x)-Q(x)|. \label{eq:dtv}
\end{align}

Following the usage in previous works, we say that a function is {\em negligible}, $f(k) = \negl(k)$, if $f(k) = o(k^{-a})$ for all $a>0$ as $k\to\infty$.

Let us formally define the worst-case decoding problem and the LPN problem. 
\begin{definition}
    The worst-case decoding problem $\wdp(n,k,w)$ is defined as follows: Given 
        \begin{enumerate}[label=\normalfont(\arabic*)]
            \item a matrix $\gG \in \ff_2^{k \times n}$
            \item  a vector $y\in \ff_2^n$ of the form $y=\gG^{\intercal} m' + e'$ for some $m' \in \ff_2^{ k}$ and $e' \in \ff_2^n$ with $|e'| = w$,
        \end{enumerate}
find an $m$ such that $y=G^\intercal m+e$ for some $e\in \ff_2^n$ with $|e| = w$. 
\end{definition}

\new{For brevity, below we write the product distribution $P_{U_k}{\times} P_{\Ber(\delta)}$ as $P_{U_k}P_{\Ber(\delta)}$.}

\begin{definition}
    The LPN problem $\lpn(k,\delta,{N},\alpha)$ with \emph{noise rate} $\delta\in (0,1/2)$, \emph{sample complexity} ${N}$, and \emph{success probability} $\alpha$ is defined as follows: Given a collection of samples $(a_i, a_i^{\intercal} m +b_i)_{i=1}^{N}$, $a_i,m\in\ff_2^k, b_i\in \ff_2$, where 
        \begin{enumerate}[label=\normalfont(\arabic*)]
            \item $m \sim P_{U_k}$ is fixed across all samples 
            \item $(a_i,b_i) \sim P_{U_k}P_{\Ber(\delta)}$ are chosen independently for each sample,
        \end{enumerate}
find $\hat m$ with $\Pr(\hat m=m)\ge \alpha$.
\end{definition}

\subsection{WDP-to-LPN reduction}\label{sec:reduction}
In this section, we discuss the reduction from the worst-case decoding problem to LPN both in terms 
of the procedure that implements it and the assumptions for the parameters of the problem
needed to prove hardness of LPN based on decoding. \new{ The reduction procedure presented below was introduced
in the  LWE case in \cite{regev2005lattices} and adapted to codes and LPN in \cite{brakerski2019worst}.}

    Let $\mathcal{A}$ be an algorithm that solves $\lpn(k,\delta,{N},\alpha)$ in time $T$. Our goal is to show that this algorithm solves the decoding problem. 
    Let $\gG$ be a $k \times n$ generator matrix $\gG$ of a code $\cC$ and a vector (noisy codeword) 
    $\gG^{\intercal}m+e$, where the noise vector $e$ satisfies $|e|= w$ 
 In order to find $m$ we proceed as follows:

\begin{enumerate}[label=\normalfont(\alph*)]
    \item Select a random vector $m'$ uniformly from $\ff_2^k$.
    \item \label{it: 2} Find a distribution $P$ on $\ff_2^n$  and a parameter $\varepsilon>0$ such that\footnote{
The probability distribution $P_{\gG Z,e^{\intercal}Z}$ is given by
    $ P_{\gG Z,e^{\intercal}Z}(x,y)=\sum_{\begin{substack}{z: \gG z=x,\\ e^\intercal z=y}\end{substack}}P_Z(z),$
and similarly for other distributions of this kind.}
    \begin{align}\label{eq: sm_red}
        \dtv(P_{\gG Z,e^{\intercal}Z},P_{U_k}P_{\Ber(\delta)}) \leq \varepsilon,
    \end{align}
    where $Z\sim P$.

    \item Generate samples $\{Z_i\}_{i=1}^{N}$ independently from $P$, and define $a_i = \gG Z_i $ and $b_i = e^{\intercal}Z_i$. Note that
    $$
    Z_i^{\intercal}(\gG^{\intercal} m' + \gG^{\intercal}m + e) = a_i^{\intercal}(m+m')+b_i.
    $$
    \item The set of pairs $(a_i,a_i^{\intercal}(m+m') + b_i)_{i=1}^{N}$ is submitted as input to Algorithm $\mathcal{A}$.
    \item If ${N}\varepsilon < \alpha$, Algorithm $\mathcal{A}$ outputs $m+m'$ with success probability at least $\alpha - {N}\varepsilon$ in time $T$.
    \item In conclusion, with probability $\alpha - {N}\varepsilon$ the message $m$ is found in time $T\!\cdot\poly(n,k)$.
\end{enumerate}   

Note that WDP is guaranteed to have a unique solution only if $w\le (\text{dist}(\cC)-1)/2$. If it does not, then the LPN problem obtained from it as described above recovers some solution to WDP with probability $\alpha-N\varepsilon$. Note also that the parameters $\delta, N,\alpha$ of the LPN problem are so far not 
related to $w$ and $n$. The interdependence of these parameters will be introduced through properties \eqref{eq:a}-\eqref{eq:b} of the reduction discussed below.

Hypothetically, if we could find $P$ such that for a vector $Z$ with $P_Z=P$,  $P_{\gG Z,e^{\intercal}Z}=P_{U_k}P_{\Ber(\delta)}$, it would imply that $\gG Z$ and $e^{\intercal}Z$ are 
distributed as a product of a uniform distribution over $\ff_2^k$ and a Bernoulli distribution. 
This ideal scenario would allow us to generate perfect inputs $(a_i, a_i(m+m')+b_i)_i$ for 
the LPN solver, ensuring that $(a_i)_i$ are uniformly distributed and independent from $(b_i)_i$.
Since this requirement generally cannot be satisfied, we relax this condition to \eqref{eq: sm_red}. This relaxation degrades the performance of the decoder from $\alpha$ to $\alpha - {N}\varepsilon$. This follows by combining \eqref{eq: sm_red} with the definition of the total variation distance \eqref{eq: dtv-max}. Indeed, let $\cE$ be the success event of the LPN solver, then the difference 
$|P_{\gG Z,e^{\intercal}Z}(\cE)-P_{U_k}P_{\Ber(\delta)}(\cE)|\le \varepsilon$, and the conclusion in Step (e) follows by the union bound.

\subsection{\new{Meaningful reductions}} The usefulness of a particular instantiation of the reduction depends on two parameters,
the distance $\dtv(P_{\gG Z,e^{\intercal}Z},P_{U_k}P_{e^{\intercal}Z})$ and the bias of the bit $e^\intercal Z$.
First, an efficient algorithm for LPN results in an efficient probabilistic decoding procedure for a generic linear code. If LPN can be solved with a low-complexity procedure
and if $\dtv(P_{\gG Z,e^{\intercal}Z},P_{U_k}P_{e^{\intercal}Z})$ is small, then the decoder
succeeds with essentially the same probability as the LPN solver.

At the same time, the smaller the bias, the harder the LPN problem becomes. Indeed, if $\bias(e^TZ)=0$, then the observation is pure noise, and so the LPN problem is essentially unsolvable. 
Judging by this boundary case, we can expect that if $\bias(e^{\intercal}Z) = \negl(k)$ the
problem is computationally hard, and hence of less theoretical interest. 
Furthermore, there are no known cryptographic implications for this parameter regime, while $\bias(e^{\intercal}Z) = O(1/\poly(k))$ implies symmetric cryptography \cite{brakerski2019worst}. 

Therefore, a {\em meaningful reduction} should satisfy the following two conditions:
\begin{subequations}
    \begin{equation}\label{eq:a}
     \dtv(P_{\gG Z},P_{U_k})<\alpha/ {N},   
    \end{equation}
    \begin{equation}\label{eq:b}
     \bias(e^{\intercal}Z) = \Omega\Big(\frac1{\poly(k)}\Big).  
    \end{equation}
\end{subequations}
Here condition \eqref{eq:a} essentially requires that random syndromes of the dual code $\cC^\bot$ be approximately uniform. It arises because
\begin{align*}
    \dtv(P_{\gG Z},P_{U_k}) 
    \leq \dtv(P_{\gG Z,e^{\intercal}Z},P_{U_k}P_{e^{\intercal}Z}) \leq \varepsilon < \alpha/{N},
\end{align*}
where the first inequality holds because for any $X,Y,Z$, 
   \begin{multline*}
   \dtv(P_{XY},P_{ZY})=\frac12\sum_{x,y}|P_{XY}(x,y)-P_{ZY}(x,y)|\\ \ge 
   \frac12\sum_x \Big|\sum_y(P_{XY}(x,y)-P_{ZY}(x,y))\Big|= \dtv(P_X,P_Z),
   \end{multline*}
 and the  
   last inequality holds because the probability of a successful reduction $\alpha-{N}\varepsilon$ must be strictly positive.  

The central question is whether objectives \eqref{eq:a}-\eqref{eq:b} can be achieved simultaneously. In this paper, we show that the answer to this 
question depends on how closely the syndromes approximate uniformity. In Section~\ref{sec:positive-rate}, we show that if $\dtv(P_{\gG 
Z},P_{U_k})=\negl(k)$, then for codes of positive rate and for constant relative error weight, the answer is negative, namely $\bias(e^{\intercal}Z)$
declines faster than any polynomial in $k$. We further show that to satisfy condition \eqref{eq:b} with $\dtv(P_{\gG Z},P_{U_k})=\negl(k)$ the error weight necessarily has to be small, namely $O(\log k)$. This shows that for codes of constant rate and small error weight, the decoding problem can be efficiently reduced to the LPN problem. Of course, since decoding from a low-weight error is easy, this reduction does not contribute to the hardness proof of LPN.

With constant error weight, the only scenario where we can expect a meaningful reduction is when $\dtv(P_{\gG Z},P_{U_k}) = \Omega(1/\poly(k))$. This case is addressed in Sec.~\ref{eq:large distance} where we show that such a reduction is achievable. At the same time, this parameter regime has an adverse effect on the accuracy of the decoder. 


\subsection{Fourier transform and Krawtchouk polynomials}
The Fourier transform of a function $f:\ff_2^n \rightarrow \rr$ is defined as follows:
\begin{equation}\label{eq:ft}
 \widehat{f}(y) = \frac{1}{2^n}\sum_{x\in {\ff_2^n}} f(x)(-1)^{x^\intercal y}, \quad y\in\ff_2^n.
\end{equation}
The inversion formula is given by   
\begin{equation}\label{eq:ift}
 f(x) = \sum_{y\in {\ff_2^n}} \widehat{f}(y)(-1)^{x^\intercal y}, \quad x\in\ff_2^n.
\end{equation}
We will also use the Fourier transform of the convolution $(f\ast g)(x):=\sum_{y} f(y)g(x-y)$, given by
\begin{equation}\label{eq:ft conv}
    (\widehat{f\ast g})({y})  = 2^n\widehat{f}({y})\widehat{g}({y}) 
\end{equation}
and Parseval's formula 
\begin{equation}\label{eq:ft Parceval}
    \frac{1}{2^n}\sum_{x}f(x)g(x) = \sum_{y}\widehat{f}({y})\widehat{g}({y}). 
\end{equation}

Denote by $S(x,w)$ the sphere centered at $x$ with radius $w$ in the Hamming space. 
The Fourier transform of the indicator function of the sphere is given by $\widehat\1_{S(0,w)}=\frac1{2^n}K_w,$
  where
    $K_w(x)= {K}_w^{(n)}(x) = \sum_{j=0}^w(-1)^j\binom{|x|}{j}\binom{n-|x|}{w-j}$ is a Krawtchouk polynomial of degree $w$. Conversely, we have
    \begin{equation}\label{ft: sphere}
     \widehat K_w(x)=\1_{S(0,w)} (x), \quad x\in \ff_2^n  .
    \end{equation}

If a function $f:\ff_2^n \rightarrow \rr$ depends only on the Hamming weight of its argument, we call it {\em radial}. Instead of $\ff_2^n,$ we sometimes write the domain of a radial function as
$\{0,1,\dots,n\}$. With some abuse of notation, we also write $f(|x|)$ instead of $f(x)$, for instance, $K_w(i)$ rather than $K_w(x)$ if $|x|=i$.

Several bounds on the magnitude of $K_t(i)$ are known in the literature. We will use the following simple universal estimate.
\begin{lemma}\label{lem: kbound}\cite[Lemma 2.4]{polyanskiy2019hypercontractivity}
    There exist constants $c\in(0,1)$ and $C\ge 1$ such that for all $0\le w \leq c n$ and $0\le i \leq n/2$,
    \begin{align}\label{eq: Kbound0}
        \frac{|K_w(i)|}{\binom{n}{w}} \leq C\Big(1-\frac{2w}{n}\Big)^i.
    \end{align}
For $n\ge 300$ it suffices to take $C=1$ and $c=0.16$.
\end{lemma}

\subsection{Code smoothing}
Here, we introduce the primary technique used in the reduction from decoding to LPN, known as code smoothing. 
\begin{definition}\label{def:smoothing}
    We say that a random vector $Z$ (or distribution $P_Z$) $\varepsilon$-smooths a code $\cC$ if $\dtv(P_{X_\cC+Z}, P_{U_n}) \leq \varepsilon$.
\end{definition}

A convenient way of characterizing the $\varepsilon$-smoothing of an $[n,k]$ linear code $\cC$ is given in the next lemma.
\begin{lemma}\label{lemma:smoothing}
    A random vector $Z$ (or a distribution $P_Z$) $\varepsilon$-smooths a code $\cC$ if and only if
    \begin{align*}
        \dtv (P_{\hH Z},P_{U_{n-k}}) \leq \varepsilon,
    \end{align*}
    where $\hH$ is the parity check matrix of $\cC$.
\end{lemma}
\begin{proof}
By a straightforward calculation,
    \begin{align*}
        P_{\hH Z}(u) = P_{\hH (X_\cC+Z)}(u) = \sum_{y: \hH y =u}P_{X_\cC+Z}(y) = 2^k P_{X_\cC+Z}(y_u),
    \end{align*}
    where $y_u \in \ff_2^n $ is any vector satisfying $\hH y_u = u$.

To prove the lemma we will show that
    \begin{align}\label{eq: smooth_eq}
         \dtv (P_{\hH Z},P_{U_{n-k}}) = \dtv(P_{X_\cC+Z}, P_{U_n}).
    \end{align}
Indeed,
    \begin{align*}
        2 \dtv (P_{\hH Z},P_{U_{n-k}})
        &= \sum_{u \in \ff_2^{n-k}}\Big|P_{\hH(X_\cC+Z)}(u)-\frac{1}{2^{n-k}}\Big|\\
        &= \sum_{u \in \ff_2^{n-k}}\Big|2^kP_{X_\cC+Z}(y_u)-\frac{1}{2^{n-k}}\Big|\\
        &= \frac{1}{|\cC|}\sum_{y \in \ff_2^{n}}\Big|2^kP_{X_\cC+Z}(y)-\frac{1}{2^{n-k}}\Big|\\
        &= \sum_{y \in \ff_2^{n}}\Big|P_{X_\cC+Z}(y)-\frac{1}{2^{n}}\Big|\\
        & = 2\dtv(P_{X_\cC+Z}, P_{U_n}).
    \end{align*}
The proof is complete. \qedhere
\end{proof}

In the next section, we argue that the reduction procedure can be formulated using the characterization given in this lemma. However, in our calculation, we primarily rely on Definition~\ref{def:smoothing} as it offers a deeper insight into the problem.

\begin{remark}\label{remark: gen math} Papers \cite{brakerski2019worst} and \cite{yu2021smoothing} used a slightly different definition of code smoothing, saying that  $Z$ smooths a code $\cC$ if $P_{\gG Z} \approx P_{U_{k}}$. It is not hard to see that this version of the definition leads to equivalent conclusions.

A characterization of smoothing related to Lemma ~\ref{lemma:smoothing} was used to quantify uniformity
guarantees for linear hashing in \cite{pathegama2024r}. It was also used in \cite{yan2024} to
establish a reduction from LPN to the average-case decoding problem.
\end{remark}

The following result is well known.

\begin{lemma}\label{lem: bias}
    Let $e \in \ff_2^n$ and $Z\sim P_Z$. The bias of the bit $e^{\intercal}Z$ is given by $2^{n-1}\hat{P}_Z(e).$
\end{lemma}
\begin{proof} We have
     \begin{align*}
        \bias(e^{\intercal}Z) &=\frac{\Pr[e^{\intercal}Z = 0]-\Pr[e^{\intercal}Z = 1]}{2}\\ 
        &=\frac{\sum_{z \in \ff_2^n }P_Z(z)\1_{\{e^{\intercal}z=0\}}-\sum_{z \in \ff_2^n }P_Z(z)
        \1_{\{e^{\intercal}z=1\}}}{2}\\
        &=\frac{\sum_{z \in \ff_2^n }P_Z(z)(-1)^{e^{\intercal}z}}{2}\\
        &= 2^{n-1}\hat{P}_Z(e).\qedhere
    \end{align*}
\end{proof}

\section{Limitations of the smoothing-based reduction}\label{sec: limit}
In this section, we show that the decoding-to-LPN reduction is feasible if the distance to uniformity $\dtv(P_{\gG Z},P_{U_k}) $ declines slowly as a function of $k$ and is not if it declines faster than $1/\poly(k)$. We start with the impossibility result.

\subsection{Positive rate, linear error weight} \label{sec:positive-rate}

The main result of this section, Thm.~\ref{thm: main_finite}, provides a bound for $\bias(e^\intercal Z)$ in terms of $\dtv(P_{\gG Z},P_{U_k}) $, 
error weight $w$, and the dual distance of the code. Using this, we show that for code sequences of growing length $n$ and positive rate,  and error weight scaling 
linearly with $n$, it is impossible to achieve a meaningful reduction if $\dtv(P_{\gG Z},P_{U_k}) = 
\negl(k)$. 

\new{The proof of Theorem~\ref{thm: main_finite} will be accomplished in several steps. Our starting point is Lemma~\ref{lem: bias}, which expresses the bias in terms of the Fourier coefficients of the smoothing distribution.
The first step of the proof, Lemma~\ref{lem: dual_bound}, a necessary condition on the Fourier coefficients for
the total variation distance to be small. Bounds on the Krawtchouk polynomials play a role here. Together, these two results
can be manipulated to yield the desired bound. 

We will begin with a technical lemma. Below $B(z,t):=\cup_{i=0}^t S(z,i)$ and $V_n(t):=|B(z,t)|=\sum_{j=0}^t\binom nj.$}

\begin{lemma}\label{lem: flatness}
    Let $\cC_0$ be a code of length $n$ and distance $d$. Let $\varepsilon>0$. Let $\rho$ be a  pmf on $\ff_2^n$ that satisfies
    \begin{align*}
        \dtv(P_{\cC_0} \ast \rho, P_{U_n}) \leq \varepsilon.
    \end{align*}
    Then 
    $$
        \sum_{x: |x|\leq t}\rho(x) 
         \leq \frac{|\cC_0|V_n({t})}{2^n}+\varepsilon
    $$
and
    $$
        \sum_{x: |x|\geq n- t}\rho(x) 
         \leq \frac{|\cC_0|V_n({t})}{2^n}+\varepsilon,
    $$
where ${t}=\lfloor\frac{d-1}{2}\rfloor$.
    
\end{lemma}

\begin{proof}
    Let us prove the first statement. Let $\cB = \bigcup_{z\in \cC_0}B(z,t)$.   Then by \eqref{eq: dtv-max}
     $$
(P_{\cC_0} \ast \rho)(\cB)-P_{U_n}(\cB)\le \dtv(P_{\cC_0} \ast \rho, P_{U_n}) \leq \varepsilon
    $$
Since the balls $B(z,t), z \in \cC_0$ are pairwise disjoint, this implies that
    \begin{align}\label{eq: isol}
        (P_{\cC_0} \ast \rho)(\cB) \leq \frac{V_n({t})|\cC_0|}{2^n}+\varepsilon.
    \end{align}
    Further
    \begin{align*}
         (P_{\cC_0} \ast \rho)(\cB)
         &=  \sum_{x \in \cB}\frac{1}{|\cC_0|}\sum_{c \in \cC_0}\rho(x-c)\\
        &=  \frac{1}{|\cC_0|}\sum_{c \in \cC_0}\sum_{x \in \cB}\rho(x-c)\\
         &\geq  \frac{1}{|\cC_0|}\sum_{c \in \cC_0}\sum_{x \in B(c,t)}\rho(x-c)\\
         &= \sum_{x: |x|\leq t}\rho(x).
    \end{align*}
Combining this with \eqref{eq: isol} yields the first statement. The proof of the second statement is completely analogous.
\end{proof}

Before proceeding to the next lemma, we introduce the following notation, which will be used throughout this section.
Let $d^\bot$ be a value of the dual distance and let $${t}^\bot=\Big\lfloor\frac{d^\bot-1}2\Big\rfloor.$$ Let
\begin{align*}
    I_{n,d^\bot} &=  \Big\{0,1,\dots,{t}^\bot\Big\} \cup \Big\{n-{t}^\bot,n-{t}^\bot+1,\dots,n\Big\}\\
    I_{n,d^\bot}^\complement &= \ion\setminus I_{n,d^\bot}.
\end{align*}

\begin{lemma}\label{lem: dual_bound}
Let $c$ and $C$ be the constants defined in Lemma \ref{lem: kbound}.    Let $\cC$ be a code of length $n$ and dual distance $d^\bot <n/2$. Let $w$ be an integer satisfying $w \leq c n$. Let $\rho$ be a pmf on $\ff_2^n$ satisfying  
    \begin{align*}
        \dtv(P_{\cC^\bot} \ast \rho, P_{U_n}) \leq \varepsilon \quad(\varepsilon>0).
    \end{align*}
    Then 
    \begin{align*}
        \frac{1}{\binom{n}{w}}\sum_{x:|x|=w}2^n\hat{\rho}(x) \leq \frac{|\cC^\bot|V_n({t}^\bot)}{2^{n-1}} + C n \Big(1-\frac{2w}{n}\Big)^{t^\bot} +2\varepsilon.
    \end{align*}
\end{lemma}

\begin{proof}
    Using \eqref{eq:ft Parceval} and \eqref{ft: sphere} and viewing $\1_{S(0,w)}$ as a function in the Fourier domain, we have
    \begin{align}
        \frac{2^n}{\binom{n}{w}}\sum_{x:|x|=w}\hat{\rho}(x) 
        & = \frac{2^n}{\binom{n}{w}}\sum_{x \in \ff_2^n}\hat{\rho}(x)\1_{S(0,w)}(x)\nonumber\\ 
        & = \frac{1}{\binom{n}{w}}\sum_{y \in \ff_2^n}\rho(y)K_w(y)\nonumber\\
         & \leq \sum_{i\in I_{n,d^\bot}}\frac{|K_w(i)|}{\binom{n}{w}}\sum_{y:|y|=i}\rho(y) 
        + \sum_{i\in I_{n,d^\bot}^\complement}\frac{|K_w(i)|}{\binom{n}{w}}\sum_{y:|y|=i}\rho(y).
        \label{eq:I}
    \end{align}
    Let us bound the first sum:
    \begin{align}
        \sum_{i\in I_{n,d^\bot}}\frac{|K_w(i)|}{\binom{n}{w}}\sum_{y:|y|=i}\rho(y) 
        &\leq \sum_{i \in I_{n,d^\bot}}\sum_{y:|y|=i}\rho(y) \nonumber\\
        &= \sum_{y:|y|        
        \leq t^\bot}\rho(y) + \sum_{y:|y| \geq n-t^\bot}\rho(y)  \nonumber\\
        & \leq \frac{2|\cC^\bot|V_n({t}^\bot)}{2^n}+2\varepsilon, \label{eq: 1st}    
    \end{align}
    where the last inequality is due to Lemma \ref{lem: flatness}. Now let us bound the second sum 
in \eqref{eq:I}.

From Lemma \ref{lem: kbound} and from the fact that $K_w(i)=(-1)^wK_w(n-i)$, for $w \leq c n$ and $i \in \ion$ we have 
\begin{align}\label{eq: Kbound}
        \frac{|K_w(i)|}{\binom{n}{w}} \leq C\Big(1-\frac{2w}{n}\Big)^{\min\{i,(n-i)\}}.
    \end{align}

Therefore, 
    \begin{align*}
        \sum_{i\in I_{n,d^\bot}^\complement}\frac{|K_w(i)|}{\binom{n}{w}}\sum_{y:|y|=i}\rho(y) 
        & \leq C\sum_{i \in I_{n,d^\bot}^\complement}\Big(1-\frac{2w}{n}\Big)^{\min\{i,(n-i)\}}\sum_{y:|y|=i}\rho(y) \nonumber\\
        & \leq C\sum_{i \in I_{n,d^\bot}^\complement}\Big(1-\frac{2w}{n}\Big)^{\min\{i,(n-i)\}} \nonumber\\
        & \leq C|I_{n,d^\bot}|\max_{i \in I_{n,d^\bot}^\complement}\Big(1-\frac{2w}{n}\Big)^{\min\{i,(n-i)\}} \nonumber\\
        & \leq C n \Big(1-\frac{2w}{n}\Big)^{t^\bot}.
     \end{align*}    
Combining this with \eqref{eq: 1st} concludes the proof.
\end{proof}

\begin{theorem}\label{thm: main_finite}
    Let $\cC$ be an $[n,k,d]$ code with dual distance $d^\bot$ and generator matrix $\gG$. Let $w$ be a non-negative integer satisfying $w \leq c n$.
    Let $Z\sim P_Z$ be a random vector defined on $\ff_2^{n}$ that satisfies 
    $\dtv(P_{\gG Z},P_{U_k}) \leq \varepsilon, \varepsilon>0$. 
    Then there exists a vector $e \in \ff_2^{n}$ of weight $w$ such that 
    \begin{align}\label{eq: main1}
        |\bias(e^\intercal Z)|  \leq \sqrt{\frac{|\cC^\bot|V_n({t}^\bot)}{2^{n+1}}} 
        + \frac{\sqrt{Cn}}{2}\Big(1-\frac{2w}{n}\Big)^{{t}^\bot/2} 
        +\sqrt{\frac{\varepsilon}{2}}.
    \end{align}
\end{theorem}
\begin{remark}
    It is not difficult to see that distributions $P_Z$ satisfying $\dtv(P_{\gG Z},P_{U_k}) \leq \varepsilon$ exist for any $\varepsilon>0$.
\end{remark}
\begin{proof}
    Setting $\rho = P_{Z} \ast P_{Z}$ and applying the data processing inequality\footnote{The data processing inequality applies to any $f$-divergence $D_f(P\|Q):={\mathbb E}_Q\Big[f\Big(\frac{dP}{dQ}\Big)
\Big]$. In particular, the total 
    variation distance between two measures is an $f$-divergence with $f(s)=\frac12|s-1|$. The inequality
says that if $P_X,Q_X$ are two measures on $X$, then $D_f(P_X\|Q_X)$ decreases if $P_X$ and $Q_X$ are ``passed''
through a conditional distribution $P_{Y|X}$ for some set $Y$. In our case $X=Y=\ff_2^n$, $P_X=P_{\cC^\bot} \ast \rho$, $Q_X=P_{U_n}$, and $P_{Y|X=x}=P_Z(\cdot-x)$.}
    \begin{align}
        \dtv (P_{\cC^\bot} \ast \rho, P_{U_n})& = \dtv( P_{\cC^\bot}\ast P_{Z}\ast P_{Z},P_{U_n}\ast P_{Z})
        \leq  \dtv (P_{\cC^\bot} \ast P_{Z}, P_{U_n}). \label{eq: slef_smooth}
    \end{align}
    Here we used the fact that $P_{U_n}=P_{U_n}\ast P_{Z}$ for any distribution $P_Z.$
    Now, from \eqref{eq: smooth_eq},
    \begin{align*}
      \dtv(P_{\cC^\bot} \ast P_Z,P_{U_{n}}) = \dtv(P_{X_{\cC^\bot} +Z},P_{U_{n}}) = \dtv(P_{\gG Z},P_{U_k}) \leq \varepsilon.
    \end{align*}
    Therefore, we have 
    \begin{align*}
        \dtv (P_{\cC^\bot} \ast \rho, P_{U_n}) \leq \varepsilon.
    \end{align*}
    From Lemma \ref{lem: dual_bound} we have 
    \begin{align}\label{eq: sum2}
         \frac{2^n}{\binom{n}{w}}\sum_{x:|x|=w}\hat{\rho}(x) 
        & \leq \frac{|\cC^\bot|V_n({t}^\bot)}{2^{n-1}} + C n \Big(1-\frac{2w}{n}\Big)^{t^\bot} +2\varepsilon.
    \end{align}
    From \eqref{eq:ft conv}, we have $\hat{\rho}(x) = 2^n\hat{P}_Z(x)^2$. Therefore,
    \begin{align*}
        \frac{1}{\binom{n}{w}}\sum_{x:|x|=w}(2^n\hat{P}_Z(x))^2 
        & \leq \frac{|\cC^\bot|V_n({t}^\bot)}{2^{n-1}} + C n \Big(1-\frac{2w}{n}\Big)^{t^\bot} +2\varepsilon.
    \end{align*}   
    Set $e = \argmin_{x: |x|=w}|2^n \hat{P}_{Z}(x)|$. Then 
    \begin{align*}
        |\bias(e^TZ)| &\ =  |2^{n-1}\hat{P}_{Z}(e)|   \hspace*{1in}(\text{Lemma }\ref{lem: bias})\\
        &\leq \frac{1}{2}\Big(\frac{1}{\binom{n}{w}}\sum_{x: |x|=w}|2^n\hat{P}_Z(x)|\Big)\\
        &\leq \frac{1}{2}\Big(\frac{1}{\binom{n}{w}}\sum_{x: |x|=w}|2^n\hat{P}_Z(x)|^2\Big)^{1/2} \quad\text{(Cauchy-Schwarz)}\\
        & \leq \frac{1}{2}\Big(\frac{|\cC^\bot|V_n({t}^\bot)}{2^{n-1}} 
        + C n \Big(1-\frac{2w}{n}\Big)^{t^\bot}
        + 2\varepsilon\Big)^{1/2}\\
        & = \Big(\frac{|\cC^\bot|V_n({t}^\bot)}{2^{n+1}} 
        + \frac{C n}{4} \Big(1-\frac{2w}{n}\Big)^{t^\bot}
        + \frac{\varepsilon}{2}\Big)^{1/2}\\
        &\leq \sqrt{\frac{|\cC^\bot|V_n({t}^\bot)}{2^{n+1}}} 
        + \frac{\sqrt{Cn}}{2}\Big(1-\frac{2w}{n}\Big)^{{t}^\bot/2} 
        +\sqrt{\frac{\varepsilon}{2}}. \qedhere
    \end{align*}
\end{proof}

\vspace*{.05in}

\new{\begin{remark} \label{remark: average}
This theorem in fact shows that not only the bias of $e^\intercal Z$ satisfies inequality \eqref{eq: main1}  for the worst error vector of weight $w$ but also the average (absolute) bias among all such error vectors satisfies this inequality. This observation implies that
even a {\em randomized} decoding-to-LPN reduction based on code smoothing is unlikely.
\end{remark}
}

Ideally, one would expect to reduce a decoding instance with linear error weight to an LPN problem instance. 
If this were possible, it would demonstrate that LPN is at least as hard as decoding a generic linear code, \new{for which the worst
case is a computationally hard problem}. However, as we will show next, this is impossible to achieve for linear codes with large dual distance unless $\dtv(P_{\gG Z},P_{U_k})$ is assumed to be very small. 

\begin{corollary}\label{cor: main2}
    Let $R \in (0,1)$ and let $\omega < c$, where $c$ is as in Lemma~\ref{lem: kbound}. 
For any sequence of $[n_j,k_j]$ linear codes $\cC_j, j=1,2,\dots$ of increasing length $n_j$ such that ${k_j}/{n_j}\to R$ 
and $d(\cC_j^\bot)/n_j\to \partial^\bot>0$ and any sequence of random vectors $(Z_{j})_{j\ge 1}$ defined on $\ff_2^{n_j}$, there exists a sequence of vectors $e_j$ with $|e_j|/n_j\to \omega$ such that the following holds true:
    \begin{enumerate}[(i)]
        \item If $\dtv(P_{\gG_jZ_j},P_{U_{k_j}}) = \negl(k_j)$, then $\bias(e_j^{\intercal}Z_j) = \negl(k_j)$ 
        \item If $\dtv(P_{\gG_jZ_j},P_{U_{k_j}}) = 2^{-\Omega(k_j)}$, then $\bias(e_j^{\intercal}Z_j) = 2^{-\Omega(k_j)}$, 
    \end{enumerate}
    where $\gG_j$ is the generator matrix of $\cC_j$.
\end{corollary}
\begin{proof} 
The bound on the right-hand side of \eqref{eq: main1} includes three terms. If the code rate stays positive as $n_j$ increases, the Hamming (sphere packing) bound implies that the term  $\sqrt{\frac{|\cC^\bot|V_n({t}^\bot)}{2^{n+1}}}$ decays exponentially with $k$. \new{This follows because for rates $R\in(0,1)$ no code sequence can attain the Hamming bound with equality since 
there are tighter bounds on the rate-distance tradeoff \cite{mceliece1977new}.} 
The second term also declines exponentially. Indeed 
\begin{equation}\label{eq:Omega}
   \frac{\sqrt{C n}}{2}\Big(1-\frac{2w}{n}\Big)^{{t}^\bot/2}
   = O\big(\sqrt{n}(1-2\omega )^{\partial^\bot n /4}\big)
   = 2^{-\Omega(n)},
\end{equation}   
where we have omitted the subscripts $j$.
   Thus the behavior of the bias for $j\to\infty$ is governed by the total variation distance $\dtv(P_{\gG_jZ_j},P_{U_{k_j}})$ as described in (i), (ii) in the statement.
\end{proof}

\begin{remark}
(a) Impossibility of reduction in Corollary \ref{cor: main2} is proved under the assumption of (normalized) error weight $\omega$ bounded by $c$. This is an artifact of the proof rather than an essential limitation. Indeed, the condition $\omega<c$ arises because we rely on a simple 
bound for Krawtchouk polynomials, the estimate \eqref{eq: Kbound}. Using more refined bounds such as those given
in \cite{KalaiLinial1995}, \cite[Lemma 2.2]{polyanskiy2019hypercontractivity} makes it possible to show that a meaningful reduction is unlikely for larger error weights as well. We omit the details.

(b) Just as the restriction on the error weight in Corollary \ref{cor: main2} is not essential, so is the assumption of the linearly growing dual distance. While typical random linear codes satisfy it
with high probability, we can in fact relax it, assuming instead that the dual distance of the codes
scales as $d(\cC_j^\bot) = \Omega(n_j^\alpha)$ for some $\alpha \in (0,1)$. Substituting this
into \eqref{eq:Omega}, we still find that the bias of $e^\intercal Z$ declines as $2^{-\Omega(n^\alpha)}$.
We note that there are abundantly many code sequences with dual distance satisfying these constraints, for instance, with high probability random codes are among them.
\end{remark}

We close this section with a remark concerning the range of parameters for the LPN reduction. We have shown that the assumption $\dtv(P_{\gG_jZ_j},P_{U_{k_j}}) = \negl(k)$
makes meaningful reductions with linear error weight impossible. This leaves the option of reducing decoding to LPN with errors of sublinear weight. We argue that if the weight satisfies $w=o(k)$, where $k$ is the dimension of the code, then a meaningful reduction is possible. At the same time, as noted earlier, this implication does not contribute toward the desired hardness proof of LPN.

According to Theorem \ref{thm: main_finite}, given that $\dtv(P_{\gG Z},P_{U_{k}}) = \negl(k)$, a necessary condition for $\bias(e^\intercal Z) = \Omega(1/\poly(k))$ is 
$$\frac{\sqrt{Cn}}{2}\Big(1-\frac{2w}{n}\Big)^{{t}^\bot/2} = \Omega(1/\poly(k)).$$ 
For codes whose dual distance $d^\bot$ scales linearly with the block length, the above condition is equivalent to $w = O(\log(k))$. We state this as another corollary whose proof is immediate.
\begin{corollary}\label{cor: main3}
     Let $R,\partial^\bot \in (0,1)$, and $l \in \integers^+$. Let $\cC_j$ be a sequence of $[n_j,k_j]$ linear codes of increasing length $n_j$ such that ${k_j}/{n_j}\to R$ and $d(\cC_j^\bot)/{n_j}\to \partial^\bot$. Denote by $\gG_j$ the generator matrix of $\cC_j$. For all $j$, let $e_j \in \ff_2^{n_j}$ be a vector satisfying $|e_j| = w_j$. If $\dtv(P_{\gG_jZ_j},P_{U_{k_j}}) = \negl(k_j)$, then the following holds true: 
        $$\text{ if }\bias(e_j^\intercal Z_j) = \Omega(1/\poly(k)),\text{ then }w_j = O(\log(k_j)).$$
\end{corollary}

It is interesting to note that the result of Corollary~\ref{cor: main3} is essentially tight. Namely,
assuming Bernoulli noise, Proposition 5.1 in \cite{debris2022worst} implies that $w = \Theta(\log(k))$ is sufficient to achieve a meaningful reduction. \new{At the same time, finding the closest codeword to the received vector in a ball of radius $w=\Theta(\log n)$ is a computationally easy task, so it has no implications for the hardness of LPN.}

\subsection{Slow Smoothing allows Reduction}\label{eq:large distance}
In Section~\ref{sec:positive-rate}, we studied the case where the error weight scales linearly with the block length, with the additional assumption that $\dtv(P_{\gG_jZ_j},P_{U_{k_j}})$ decays fast. Here we show that
lifting this assumption makes it possible to obtain a meaningful reduction.

\begin{theorem}\label{prop: ach}
    Let $R \in (0,1)$, $\omega \in (0,1/2)$, and $l \in \integers^+$. Let $\cC_j$ be a sequence of $[n_j,k_j]$ linear codes of increasing length $n_j$ such that ${k_j}/{n_j}\to R$. Let $\gG_j$ be a generator matrix of $\cC_j$ and let $e_j \in \ff_2^{n_j}$ be a vector satisfying $|e_j| = \lfloor\omega n_j\rfloor$. Then there exists a sequence of distributions $P_{Z_j}$ satisfying the following conditions:
    \begin{enumerate}[label=\normalfont(\roman*)]
        \item $\dtv(P_{\gG_j Z_j, e_j^{\intercal}Z_j},P_{U_{k_j}}P_{e_j^{\intercal}Z_j}) = O(k_j^{-l})$,
        \item $\bias(e_j^\intercal Z_j ) = \Omega(k_j^{-l})$.
    \end{enumerate}
\end{theorem}
This theorem follows directly from the next lemma upon setting $\gamma=k_j^{-l}$.
\begin{lemma}\label{lemma: distance-bias}
    Let $\cC\in \ff_2^n$ be linear code satisfying $|\cC| < {2^{n-1}}/V_n(2)$ and let 
$e \in \ff_2^n\setminus\{0\},|e|<n/2$. Then for any $\gamma >0$, the following conditions are simultaneously achievable:
    \begin{enumerate}[label=\normalfont(\roman*)]
        \item $\dtv(P_{\gG Z, e^{\intercal}Z},P_{U_k}P_{e^{\intercal}Z}) \leq \gamma \big(\frac{3}{2}-\frac{|e|}{n}\big)$,
        \item $\bias(e^{\intercal}Z) =\frac{\gamma}{2}\big(1-\frac{2|e|}{n}\big)$.
    \end{enumerate}
\end{lemma}

\begin{proof}
    Let $P_Z = (1 -\gamma)P_{U_n} + \gamma \frac{\1_{S(0,1)}}{n}$. 
By Lemma \ref{lem: bias} and the definition of $K_w(i)$
      $$
    \bias(e^{\intercal}Z) = 2^{n-1} \hat{P}_Z(e) =\frac12\Big((1-\gamma)\1_{\{0\}}(e)+\frac \gamma nK_1(|e|)\Big)\stackrel{(e\ne 0)}=
    \frac{\gamma }{2}\Big(1-\frac{2|e|}{n}\Big).
    $$
Denote by $\cC_e$ the code generated by 
$\gG_e:=\stirlingI\gG{e^{\intercal}}$and note that $P_{\gG Z, e^{\intercal}Z}=P_{\gG_eZ}$. By the triangle inequality,
    \begin{align*}
        \dtv(P_{\gG_eZ},P_{U_k}P_{e^{\intercal}Z}) \leq \dtv(P_{U_{k+1}},P_{U_k}P_{e^{\intercal}Z}) + \dtv(P_{\gG_e Z},P_{U_{k+1}}).
    \end{align*}
To bound the first term, use \eqref{eq:dtv} to write
    \begin{align*}
        \dtv(P_{U_{k+1}},P_{U_k}P_{e^{\intercal}Z}) 
        & = \frac{1}{2}\sum_{x \in \ff_2^k,b \in \ff_2}|P_{U_{k+1}}(x,b)-P_{U_k}(x)P_{e^{\intercal}Z}(b)|\\
        & = \frac{1}{2}\sum_{x \in \ff_2^k}\big|\frac{1}{2^{k+1}}-\frac{1}{2^k}P_{e^{\intercal}Z}(0)\big| 
        + \frac{1}{2}\sum_{x \in \ff_2^k}\big|\frac{1}{2^{k+1}}
        - \frac{1}{2^k}P_{e^{\intercal}Z}(1)\big|\\
        &= \bias(e^{\intercal}Z)\\
        &= \frac{\gamma }{2}\big(1-\frac{2|e|}{n}\big).
    \end{align*}
    
    Now we are left to bound $\dtv(P_{\gG_e Z},P_{U_{k+1}})$. 
    From \eqref{eq: smooth_eq}, we have
    \begin{align*}
        \dtv(P_{\gG_e Z},P_{U_{k+1}}) = \dtv(P_{X_{\cC_e}+ Z},P_{U_{n}}).
    \end{align*}
    
We divide the analysis into two cases depending on the value of $d^\bot(\cC_e)$. First assume that
$d(\cC_e^\bot) \geq 3$. In this case, we obtain 
    \begin{align*}
        P_{X_{\cC_e^\bot}+Z} (x) = (1-\gamma)P_{U_n}(x) + \gamma \Big(\frac{\1_{S(0,1)}}{n} \ast P_{\cC_e^\bot}\Big)(x) = \frac{1-\gamma}{2^n}+\frac{\gamma \1_{\cS}(x)}{n |\cC_e^\bot|}, 
    \end{align*}
    where $\cS = \bigcup_{z \in \cC_e^\bot} S(z,1)$. Therefore,
    \begin{align*}
        \dtv(P_{X_{\cC_e^\bot}+Z}, P_{U_n}) 
        &= \frac{1}{2}\sum_{x \in \ff_2^n}\Big|(P_Z\ast P_{\cC_e^\bot})(x)-\frac{1}{2^n}\Big|\\
        &=\frac12\Big[\sum_{x \in \cS}\Big| \frac{1-\gamma}{2^n} 
        + \frac{\gamma}{n |\cC_e^\bot|}
        -\frac{1}{2^n}\Big|
        + \sum_{x \in \ff_2^n \setminus \cS}\Big| 
        \frac{1-\gamma}{2^n}-\frac{1}{2^n}\Big|\Big]\\
        &=\frac12\Big[n|\cC_e^\bot|\gamma\Big(\frac{1}{n|\cC_e^\bot|} -\frac{1}{2^n}\Big)
        + (2^n - n|\cC_e^\bot|)\frac{\gamma}{2^n}\Big]\\
        & =\gamma\Big(1-\frac{n|\cC_e^\bot|}{2^n}\Big)\\
        & \leq \gamma.
    \end{align*}
 
Now assume that $d(\cC_e^\bot) \in \{1,2\}$. Since $|\cC_e| \leq 2^n/V_n(2)$, $|\cC_e^\bot|>V_n(2)$, so 
there exists a nontrivial subcode $\cC^\prime\subset\cC_e^\bot$ such that $d(\cC^\prime)\geq 3$. By the 
above argument we have
    \begin{align*}
        \dtv(P_{X_{\cC^\prime}+Z}, P_{U_n}) \leq \gamma.
    \end{align*}
    Let $\mathcal{L}$ be the set of coset leaders in $\cC_e^\bot/\cC^\prime$.
    Then we can write $ P_{\cC_e^\bot} = \frac{1}{|\cC_e^\bot/\cC^\prime|}\sum_{l \in \mathcal{L}} P_{\cC^\prime + l}$. 
From the convexity of the total variation distance, we have
    \begin{align*}
        \dtv(P_{X_{\cC_e^\bot}+Z}, P_{U_n}) \leq \frac{1}{|\mathcal{L}|}\sum_{l \in \mathcal{L}}\dtv(P_{X_{\cC^\prime+l}+Z}, P_{U_n}) \leq \gamma.
    \end{align*}
    Therefore, $\dtv(P_{\gG Z, e^{\intercal}Z},P_{U_k}P_b) \leq \gamma + \frac{\gamma }{2}\big(1-\frac{2|e|}{n}\big) = \gamma \big(\frac{3}{2}-\frac{|e|}{n}\big)$.
Since $|e|<n/2$, in both cases condition (i) is satisfied. Thus, the proof is complete.
\end{proof}

\new {\begin{remark}\label{remark: nonbinary}
As pointed out by a reviewer, the results of this subsection extend without difficulty to nonbinary 
codes. We briefly outline the changes needed for this extension. Let $\Z_q$ be the additive group of integers
mod $q$. We consider additive codes over $\integers_q$, i.e., additive subgroups of $\Z_q^n$. Define 
a random variable $\xi$ with
$$
   \Pr(\xi=0)=1-\delta, \quad \Pr(\xi=s)=\frac{\delta}{q-1}, \;s\in \Z_q\backslash 0,
   $$
where $\delta\in(0,1)$. The noise vector $Z$ is formed of $n$ independent copies of $\xi$.
The definition of the bias is modified as follows:
  $$
  \bias(\xi)=\frac 1q-\frac{\delta}{q-1},
  $$
and as before, $\bias(e^\intercal Z)=q^{n-1}\hat P_Z(e)$, where the Fourier transform of a 
function $f:\integers_q\to {\mathbb R}$ is defined as
     $$
     \hat{f}(y)=\frac{1}{q^n}\sum _{x\in {\mathbb {Z}}_q^n} f(x)\exp\Big(\frac{2\pi i }{q}x^\intercal y\Big).
     $$
The statement of Lemma~\ref{lemma: distance-bias} is modified as follows.
\begin{lemma} Let $\cC\subset \Z_q^n$ be an additive code satisfying $|\cC|<q^{n-1}/V_n(2)$
and let $e\in Z_q^n\backslash\{0\}$, $|e|<(q-1)n/q$. Then for any $\gamma>0$ the following
conditions are simultaneously achievable:
\begin{enumerate}[label=\normalfont(\roman*)]
    \item $\dtv(P_{\gG Z, e^\intercal Z}, P_{U_k}P_{e^\intercal Z})\le \gamma(2-\frac 1q-\frac{|e|}n)$,
    \item $\bias(e^\intercal Z)=\gamma(\frac 1q-\frac{|e|}{(q-1)n})$.
\end{enumerate}
\end{lemma}
The proof starts with the distribution $P_Z = (1 -\gamma)P_{U_n} + \gamma \frac{\1_{S(0,1)}}{(q-1)n}$ and
proceeds with minimal changes (recall that $\hat\1_{S(0,1)}(x)=K_1(x):=(q-1)n-q|x|$).

\vspace*{.1in}Note that the main results concerning the impossibility of the reduction, given in Sec.~\ref{sec:positive-rate}, do not extend in the same manner because they critically depend on the estimate of
Krawtchouk polynomials \eqref{eq: Kbound0}, which does not have an immediate $q$-ary analog.

\end{remark}
}
\new{\section{Concluding remarks}
Let us summarize the results of this work with regard to the decoding-to-LPN reduction based on code smoothing. Previous works 
\cite{brakerski2019worst,yu2021smoothing,debris2022worst} considered the case of zero-rate codes, showing that an efficient reduction
was possible. Expanding on these results, we examined the case $0<R<1$, isolating the parameter regimes for which such
a reduction is possible and the regimes for which it is not. In particular, a decoding-to-LPN reduction is not possible if the 
distribution of the syndromes of the dual code approaches a uniform distribution, namely the distance $\dtv(P_{\gG Z}, P_{U_k})$ 
decays faster than any inverse polynomial function of the code dimension $k$. At the same time, a reduction is possible if this 
distance behaves as $\Omega(\frac 1{\poly(k)})$, although in this case what we show is that LPN is as hard as
the worst-case decoding problem that operates with a non-negligible error rate.

Looking ahead, it would be of interest to extend the impossibility results of Sec.~\ref{sec:positive-rate} to nonbinary codes and to derive sharper estimates of the parameter regimes for which the reduction fails in the case of lattices and the Learning-with-Errors problem compared to \cite{regev2005lattices,debris2023smoothing}.}

\section*{Acknowledgments}
\new{We are grateful to the anonymous reviewers for a careful reading of our first draft and insightful comments. In particular,
Remarks \ref{remark: average} and \ref{remark: nonbinary} added in the final version are due to a reviewer.}

This research was partially supported by NSF grants CCF-2110113 (NSF-BSF), CCF-2104489, and CCF-2330909. This work was partly done when the authors were visiting the Simons Institute for the Theory of Computing in Spring 2024. The authors are grateful to Nicolas Resch for a useful 
discussion of the decoding-to-LPN reduction problem.

\bibliographystyle{abbrvurl}
\bibliography{smoothing}
\end{document}